\documentclass[12pt]{article}

\usepackage{geometry}  
\usepackage{graphicx}
\usepackage{amsthm}
\usepackage{amsmath}
\usepackage{amssymb}
\usepackage{hyperref}
\usepackage{enumerate}
\usepackage{enumitem}

\newtheorem{theorem}{Theorem}

\newtheorem{prop}{Proposition}

\usepackage{xcolor}

\newcommand{\sgn}{\operatorname{sgn}}

\title{Rotationally Symmetric Zoll Metrics\\with a Cubic Integral}

\author{
Holger R. Dullin\thanks{
School of Mathematics and Statistics, University of Sydney, Camperdown NSW 2050, Australia.\\
Emails:
\texttt{holger.dullin@sydney.edu.au},
\texttt{gleb.palshin@sydney.edu.au}
}
\quad
Vladimir S. Matveev\thanks{
Institut f\"{u}r Mathematik, Friedrich Schiller Universit\"{a}t Jena, 07737 Jena, Germany.\\
Emails:
\texttt{vladimir.matveev@uni-jena.de},
\texttt{serena.scapucci@uni-jena.de}
}
\\[0.7em]
Gleb P. Palshin\footnotemark[1]
\quad
Serena Scapucci\footnotemark[2]
}

\date{}

\begin{document}

\maketitle
\vspace{-1.5em}

\begin{abstract}
    We introduce a new method for constructing Zoll metrics that admit first integrals polynomial in the momenta. The method uses Funk's form of the metric together with an algebraic relation between the first integrals. This reduces the problem to an algebraic system from which all unknown functions can be determined explicitly. In the case of a non-trivial cubic integral, we present an explicit two-parameter family of metrics that covers all such rotationally symmetric Zoll metrics on the $2$-sphere: every such metric is isometric to a member of this family. We compare these metrics with the earlier results of Valent, Duval, and Shevchishin.
\end{abstract}

\section{Introduction}\label{sec1}

All geodesics of a two-dimensional closed Riemannian manifold $(\mathcal{M}^2,g)$ can be closed only if $\mathcal{M}^2$ is homeomorphic to the real projective plane $\mathbb{R}P^2$ or to the sphere $\mathbb{S}^2$, see~\cite[Theorem~7.37]{Besse1978}. The former case implies that the curvature of the metric $g$ is constant, see~\cite[Theorem~2.6]{Pries2009}. In the latter case, the first non-trivial examples were found by Zoll \cite{Zoll1903}. Since then, any Riemannian manifold $(\mathcal{M}^2,g)$ such that \mbox{$\mathcal{M}^2 \cong \mathbb{S}^2$} and all of whose geodesics are closed is called a {\it Zoll surface}, and the corresponding metric $g$ is called a {\it Zoll metric}. In particular, all geodesics of a Zoll metric necessarily have the same length, see~\cite{GromollGrove1981}. More detailed information on manifolds all of whose geodesics are closed can be found in Besse's monograph~\cite{Besse1978}. Classical results about such manifolds are due to Bott~\cite{Bott1954} and Weinstein~\cite{Weinstein1974}.  Guillemin~\cite{Guillemin1976} classified infinitesimal Zoll deformations of the round metric by an odd function on the sphere. Projective Zoll structures are analysed in
\cite{LeBrunMason2002}, \cite{LeBrunMason2010}, 
and \cite{MazzucchelliSuhr2019} gives an intrinsic characterisation of Zoll metrics.

For a surface of revolution Zoll metrics can be written down explicitly in a form due to Funk~\cite{Funk1913}. 
In Funk's metric there is functional freedom of choice of an odd function $h(x)$ with certain boundary conditions, see below. The resulting metric has a linear integral corresponding to the rotational symmetry of the surface of revolution. Whether there are additional higher degree polynomial integrals depends on the choice of $h$. 
Any admissible choice of $h$ does give a Zoll metric, however, in general these do not have a \emph{polynomial} third integral, so are not polynomially superintegrable. 
General results about metrics with polynomial integrals are due to Kolokoltsov \cite{Kolokoltsov1983}. Kiyohara \cite{Kiyohara2001} showed that metrics with high degree polynomial integrals exist, but recently Matveev \cite{Matveev2025} showed that these examples while they are Zoll metrics they are not polynomially superintegrable.

It is known that Funk metrics cannot have an additional quadratic integral, see~\cite{Kiyohara1991}. 
Recently polynomially superintegrable Zoll metrics have become an active area of research.
The case of the cubic integral has been described by Matveev and Shevchishin
\cite{MatveevShevchishin2011}.
Based on this work explicit forms for such Zoll metrics were given by Valent for the cubic case
\cite{Valent2014}, \cite{ValentDuvalShevchishin2015}.
 We discuss these result in detail in section~\ref{sec3} of this paper. The existence of such metrics with a quartic integral was discovered by Novichkov \cite{Novichkov2015}. It is also known that for any odd degree of the additional integral, there exists a multi-parameter family of Zoll metrics obtained by Valent in~\cite{Valent2021}.

The general classification of polynomially superintegrable Zoll metrics remains incomplete. Since the typical approach to the problem involves solving nonlinear ODEs, the complexity increases significantly with the degree of the integral. 
In this paper, we show how this can be avoided for surfaces of revolution by reducing the problem to a system of algebraic equations. We demonstrate the effectiveness of this method for the cubic case and are able to give the complete list of all Funk metrics with a cubic integral. We believe that this approach can be generalised to the case of Funk's metrics with an integral of arbitrary degree. We also note that the result obtained here is consistent with the more general idea of finding superintegrable geodesic flows with polynomial integrals in an algebraic way, recently expressed in~\cite{Matveev2025}.

\subsection{Funk Metrics}\label{sec11}

In~\cite{Funk1913}, Funk found the following form of a rotationally symmetric Zoll metric
$$
    g = (1+h(\sin r))^2 dr^2 + \cos^2r\,dy^2,
$$
where 
$\displaystyle{r \in \left[-\frac{\pi}{2}, \frac{\pi}{2}\right]}$, $y \in \mathbb{S}^1$,
and the following conditions are satisfied:
\begin{enumerate}[label=(F\arabic*)]
    \item the function $r \mapsto h(\sin r)$ is real-valued for all $\displaystyle{r \in \left[-\frac{\pi}{2},\frac{\pi}{2}\right]}$, that is, $h \colon [-1, 1] \to \mathbb{R}$;\label{condition:one}
    \item the function $h$ is smooth, that is~$h \in C^\infty[-1,1]$;\label{condition:two}
    \item the function $h$ is odd, that is~$h(-x) = -h(x)$ for all $x \in [-1, 1]$;\label{condition:three}
    \item at the endpoints, one has $h(1) = h(-1) = 0$;\label{condition:four}
    \item $h(x) + 1 > 0$ for all $x \in [-1, 1]$.\label{condition:five}
\end{enumerate}
Moreover, by \cite[Theorem~4.13 and Corollary~4.16]{Besse1978}, every rotationally symmetric Zoll metric on $\mathbb{S}^2$, up to isometry and normalisation, can be written in this form.
After the coordinate change $\sin r \to x$, the Funk metric is
\begin{equation}\label{eq1}
    g = \frac{(h(x)+1)^2}{1-x^2} dx^2 + (1-x^2) dy^2,
\end{equation}
where $h(x)$ satisfies conditions~\ref{condition:one}--\ref{condition:five}.

\subsection{First Integrals Polynomial in Momenta}\label{sec12}

As a mechanical system, the geodesic flow of the metric~\eqref{eq1} has two obvious constants of motion, namely, the Hamiltonian
\begin{equation}\label{eq2}
    H = \frac{1}{2}\left(
        \frac{1-x^2}{\left(h(x)+1\right)^2}p_x^2 +
        \frac{1}{1-x^2} p_y^2
    \right)
\end{equation}
quadratic in the momenta $(p_x,p_y)$, representing the kinetic energy of the system, and hence the speed of the geodesic, and the linear first integral associated with the rotational symmetry
\begin{equation}\label{eq3}
    L = p_y.
\end{equation}
The closedness of all geodesics implies maximal superintegrability of the geodesic flow
$$
    \frac{d\boldsymbol{q}}{dt}=\frac{\partial H}{\partial \boldsymbol{p}},
    \quad
    \frac{d\boldsymbol{p}}{dt}=-\frac{\partial H}{\partial \boldsymbol{q}},
    \qquad
    \boldsymbol{q}=(x,y),
    \quad
    \boldsymbol{p}=(p_x,p_y),
$$
and hence the presence of an additional non-trivial integral $F$, functionally independent of $H$ and $L$ (not necessarily polynomial in the momenta).

In this paper, we study the special case in which there exists a first integral $F$ cubic in the momenta $(p_x,p_y)$, that is
\begin{equation}\label{eq4}
    F := \sum_{k=0}^{3} A_k(x,y) p_x^{3-k} p_y^k.
\end{equation}
The problem then is to find all possible functions $h(x)$ and $A_k(x,y)$ such that \mbox{$\{ H, F \} \equiv 0$}, where $\{ \cdot,\cdot \}$ is the standard Poisson bracket defined as
\begin{equation}\label{eq5}
    \{ f_1, f_2 \} := \left( \frac{\partial f_1}{\partial x} \frac{\partial f_2}{\partial p_x} - \frac{\partial f_1}{\partial p_x} \frac{\partial f_2}{\partial x} \right) +
    \left( \frac{\partial f_1}{\partial y} \frac{\partial f_2}{\partial p_y} - \frac{\partial f_1}{\partial p_y} \frac{\partial f_2}{\partial y} \right)
\end{equation}
for arbitrary smooth functions $f_1$ and $f_2$ on $T^*\mathcal{M}^2$.

\subsection{Main Result}\label{sec13}

We prove the following theorem that explicitly presents all rotationally symmetric Zoll metrics with a cubic integral.

\begin{theorem}\label{t1}
Let $g$ be a rotationally symmetric smooth Zoll metric in Funk's form~\eqref{eq1} whose geodesic flow admits a first integral that is non-trivially cubic in the momenta. Then the function $h(x)$ is given by
\begin{equation}\label{eq6}
    h(x) = 
    x \sqrt{\frac{Q_1(x) - 2\sqrt{Q_2(x)}}{Q_2(x)}},
\end{equation}
where
\begin{align*}
    Q_1(x) &= -(2\sigma_1+1)x^2 + 2\sigma_1 + 3,
    \\
    Q_2(x) &= (2\sigma_1 + 4\sigma_2 + 1)x^4 - (6\sigma_1 + 8\sigma_2 +3)x^2 + 4\sigma_1 + 4\sigma_2 + 3,
\end{align*}
for some real parameters $\sigma_1$ and $\sigma_2$ satisfying
\begin{equation}\label{eq7}
    \sigma_1 > -\frac{3}{2}
    \quad
    \text{and}
    \quad
    -\left(\sigma_1+\frac{3}{4}\right) 
    < \sigma_2 <
    \frac{(2\sigma_1+1)(2\sigma_1-3)}{16}.
\end{equation}
Conversely, for every pair $(\sigma_1,\sigma_2)$ satisfying~\eqref{eq7}, the metric~\eqref{eq1}, with $h(x)$ given by~\eqref{eq6}, is a rotationally symmetric Zoll metric whose geodesic flow admits a non-trivial cubic integral.
\end{theorem}
\begin{proof}
    See Section~\ref{sec2} of this paper.
\end{proof}

Here, by non-triviality, we mean the functional independence of the cubic integral from any integrals of lower degree in momenta. The set of admissible values of the parameters $\sigma_1$ and $\sigma_2$ from Theorem~\ref{t1} is shown in Figure~\ref{fig1}.

\begin{figure}
    \centering
    \includegraphics[width=0.4\linewidth]{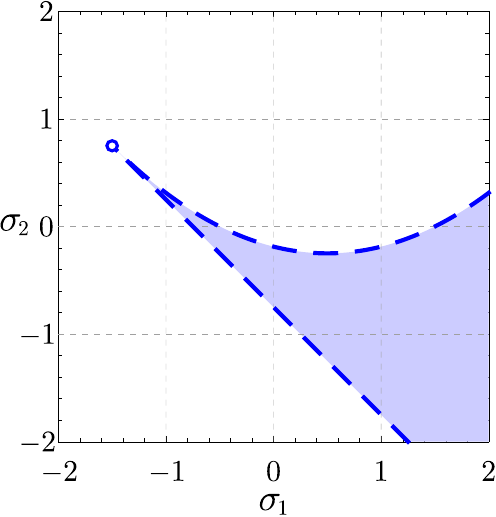}
    \caption{Admissible parameter values (shaded region) for rotationally symmetric Zoll metrics admitting a non-trivial cubic integral}
    \label{fig1}
\end{figure}

\section{Proof of Theorem~\ref{t1}}\label{sec2}

We start with the approach that was successfully applied in~\cite{MatveevShevchishin2011}. Let us define a linear operator
\begin{equation}\label{eq8}
    \mathcal{L} \colon F \mapsto \{ F, L \} 
\end{equation}
acting on the space of complex integrals of degree $3$ in momenta. The first integrals of degree $3$, generated by~\eqref{eq2} and~\eqref{eq3}, namely,
$$
    H L
    \quad
    \text{and}
    \quad
    L^3,
$$
are eigenvectors of~\eqref{eq8} with eigenvalue $0$. Since the space of cubic integrals is finite-dimensional (see~\cite{Kruglikov2008}), there are two possibilities: either there exists a non-zero eigenvalue $\alpha \in \mathbb{C}$ or the spectrum of $\mathcal{L}$ consists only of zero.

\subsection{The Case of a Non-Zero Eigenvalue}\label{sec21}

By solving $\mathcal{L}(F)=\alpha\,F$ for complex $\alpha \neq 0$ and $F$ of the form~\eqref{eq4}, we obtain
\begin{equation}\label{eq9}
    F_+ := e^{\alpha y}\sum_{k=0}^3 a_k(x) p_x^{3-k} p_y^k
\end{equation}
with unknown functions $a_k(x)$. In addition, the operator $\mathcal{L}$ has another non-zero eigenvalue $(-\alpha)$ corresponding to the eigenfunction
\begin{equation}\label{eq10}
    F_- := e^{-\alpha y}\sum_{k=0}^3 a_k(x) p_x^{3-k} (-p_y)^k.
\end{equation}
Moreover, if~\eqref{eq9} is a first integral, then~\eqref{eq10} is also a first integral, since the systems of equations produced by the monomials in the vanishing brackets $\{ H, F_+ \}$ and $\{ H, F_- \}$ have the same set of solutions. The existence of the second cubic integral is due to the presence of symmetry with respect to the sign change $y \to -y$.

Let us construct a degree~$6$ integral
\begin{equation}\label{eq11}
    P := F_+ F_-.
\end{equation}
Note that $P$ is independent of $y$, and hence $\{ L, P \}=0$. Thus, the triple of first integrals $(H, L, P)$ pairwise commutes with respect to the Poisson bracket~\eqref{eq5}.
Since $\mathrm{dim}(\mathcal{M})=2$ is the maximum number of independent integrals in involution, the triple is functionally dependent.
Let us also define a degree $5$ integral
$$
    K := \{ F_+, F_- \}.
$$
Since
$$
    \{ L, K \} = \alpha \{ F_+, F_- \} - \alpha \{ F_+, F_- \} = 0,
$$
the triple $(H,L,K)$ is pairwise commuting and $K = K(H,L)$. Hence, the integrals $(H, L, F_+, F_-)$ form a Poisson algebra given by the matrix
$$
    \Pi =
    \begin{bmatrix}
    0 & 0 & 0 & 0
    \\
    0 & 0 & -\alpha F_+ & \alpha F_-
    \\
    0 & \alpha F_+ & 0 & K(H,L) 
    \\
    0 & -\alpha F_- & -K(H,L) & 0
    \end{bmatrix}.
$$
We know that the integral $F_+ F_- = P(H,L)$ is also a function of $(H, L)$, so $(H,L,F_+,F_-)$ are bounded by the following constraint:
$$
    G \equiv 0,
    \qquad
    G := F_+ F_- - P(H,L).
$$
Since $G$ is constantly zero as a function on the phase space, it commutes with any other function defined on it, including $H$, $L$, $F_+$, and $F_-$. Hence, $G$ is a Casimir of $\Pi$, i.\,e.~$\nabla G \in \ker \Pi$, where $\nabla$ is the gradient with respect to $(H,L,F_+,F_-)$. For non-zero (non-trivial) integrals $F_+$ and $F_-$ one obtains
$$
    \Pi \cdot \nabla G = \overline{0} 
    \quad
    \iff 
    \quad
    K - \alpha \frac{\partial P}{\partial L} = 0.
$$
Therefore,
$$
    K(H,L) = \alpha \frac{\partial P(H,L)}{\partial L}.
$$
The Poisson structure then is
$$
    \Pi = 
    \alpha
    \begin{bmatrix}
    0 & 0 & 0 & 0
    \\
    0 & 0 & - F_+ & F_-
    \\
    0 & F_+ & 0 & \partial_L P 
    \\
    0 & - F_- & -\partial_L P & 0
    \end{bmatrix},
$$
and $F_+ F_- - P(H,L)$ is the Casimir of this Poisson structure. The following statement provides an explicit formula for this Casimir.

\begin{prop}\label{prop1}
    The first integral $P$ has the form
    \begin{equation}\label{eq12}
        P = c_0 H^3 + c_1 H^2 L^2 + c_2 H L^4 + c_3 L^6
    \end{equation}
    for some constants $c_0, \dots, c_3$.
\end{prop}
\begin{proof}
    Let us assume $L \neq 0$ and introduce a parameter $t \in \mathbb{R}$, such that $p_x = t\,p_y=t\,L$. Then,
    $$
        H = \frac{L^2}{2(1-x^2)}\left[
        \left(\frac{1-x^2}{1+h(x)}\right)^2 t^2 +
        1
    \right]
    $$
    which can be solved for $t^2$ as
    \begin{equation}\label{eq13}
        t^2 = \frac{(1+h(x))^2}{1-x^2} \left(
            \frac{2H}{L^2} - \frac{1}{1-x^2}
        \right).
    \end{equation}
    By definition~\eqref{eq11},
    \begin{align}
    \label{eq14}
        P 
        &:= F_+ F_- 
        = e^{\alpha\,y} e^{-\alpha\,y} \left( \sum_{k=0}^3 a_k t^{3-k} L^3 \right) \left( \sum_{k=0}^3 (-1)^k a_k t^{3-k} L^3 \right)
        \nonumber \\
        &= L^{6}\left[
            a_0^2 t^6 
            - (a_1^2 - 2 a_0 a_2) t^4
            + (a_2^2 - 2 a_1 a_3) t^2
            - a_3^2
        \right],
    \end{align}
    which depends only on even powers of $t$. Substituting~\eqref{eq13} into~\eqref{eq14}, we obtain
    a quasi-homogeneous polynomial in $(H,L)$ with weights $(2,1)$ and coefficients depending on $x$. Let us denote these coefficients by
    $$
        c_0(x),
        \quad
        c_1(x),
        \quad
        c_2(x),
        \quad
        \text{and}
        \quad
        c_3(x),
    $$
    so that
    $$
        P = \sum_{k=0}^{3} c_k(x) H^k L^{6-2k}.
    $$
    Initially, we assumed that $L$ is non-zero, however, this identity
    extends to $L=0$ by polynomiality.
    Finally, taking into account $\{P,H\} \equiv 0$,
    \begin{align*}
        0 &= \{ P, H \} 
        = \sum_{k=0}^{3} \left\{ c_k(x) H^k L^{6-2k}, H \right\}
        \\
        &= \sum_{k=0}^{3} \left( 
        H^k L^{6-2k} \left\{ c_k(x), H \right\} +
        c_k(x) \left\{ H^k L^{6-2k}, H \right\}
        \right).
    \end{align*}
    Since the second Poisson bracket in the summation is equal to zero,
    $$
        0 
        = \sum_{k=0}^{3} H^k L^{6-2k} \left\{ c_k(x), H \right\}
        = \frac{1-x^2}{(1+h(x))^2}\,p_x \sum_{k=0}^{3} H^k L^{6-2k} \,c^\prime_k(x),
    $$
    where $c^\prime_k = \frac{d}{dx}c_k$.
    For the open set where $p_x \neq 0$, $p_y\neq0$, and $x \in (-1,1)$, the factor
    $$
        \frac{1-x^2}{(1+h(x))^2}\,p_x
    $$
    is non-zero, so we get
    $$
        \sum_{k=0}^{3} c^\prime_k(x) H^k L^{6-2k} = 0.
    $$
    For fixed $x$, the functions $H^k L^{6-2k}$ are linearly independent as homogeneous polynomials in $(p_x,p_y)$, which gives
    $$
        c^\prime_k(x)=0,
        \qquad
        k=0,\dots,3.
    $$
    Hence each $c_k$ is constant.
    This proves the proposition.
\end{proof}

Thus, we have obtained an explicit form of the polynomial dependence of the first integrals, the existence of which was proved in~\cite{Matveev2025}.
Note that the Casimir is invariant under canonical coordinate transformations and therefore is convenient for comparing superintegrable geodesic flows written in different coordinates, which will be demonstrated in \mbox{Section~\ref{sec3}}. Furthermore, the explicit form of the Casimir, revealed by Proposition~\ref{prop1}, provides an algebraic relationship between the first integrals and hence the undetermined functions in their coefficients. The constant coefficients of the Casimir serve as parameters of the superintegrable system.

\begin{prop}\label{prop2}
    The functions $h(x), a_0(x), \dots, a_3(x)$ in~\eqref{eq2},~\eqref{eq9} and~\eqref{eq10} satisfy the following system of algebraic equations:
    \begin{align}\label{eq15}
        a_0^2(x) &= \frac{c_0(1-x^2)^3}{8(h(x)+1)^6},
        \nonumber
        \\
        a_1^2(x) - 2 a_0(x) a_2(x) &= 
        -\frac{
            \left(1-x^2\right) \left(
                3 c_0
                + 2 c_1 \left(1-x^2\right)
            \right)
        }{
            8 (h(x)+1)^4
        },
        \nonumber
        \\
        a_2^2(x) - 2a_1(x)a_3(x) &=
        \frac{
            3 c_0
            + 4 c_1 \left(1-x^2\right)
            + 4 c_2 \left(1-x^2\right)^2
        }{
            8 \left(1-x^2\right) (h(x)+1)^2
        },
        \\
        a_3^2(x) &=
        -\frac{
            c_0
            + 2 c_1 \left(1-x^2\right)
            + 4 c_2 \left(1-x^2\right)^2
            + 8 c_3 \left(1-x^2\right)^3
        }{
            8 \left(1-x^2\right)^3
        },
        \nonumber
        \\
        a_3(x) &= \frac{x}{\alpha(1-x^2)} a_2(x),
        \nonumber
    \end{align}
where $c_0$, $c_1$, $c_2$, and $c_3$ are real parameters defined in Proposition~\ref{prop1}.
\end{prop}
\begin{proof}
    By Proposition~\ref{prop1},
    $$
        F_+ F_- - \left(c_0 H^3 + c_1 H^2 L^2 + c_2 H L^4 + c_3 L^{6}\right) = 0.
    $$
    Direct substitution of~\eqref{eq2}, \eqref{eq3}, \eqref{eq9}, and~\eqref{eq10} into this equation makes its left-hand side a homogeneous polynomial of degree $6$ in $(p_x, p_y)$. The coefficients of its monomials have to vanish simultaneously, which gives the first four equations in~\eqref{eq15}.

    To obtain the last equation, one has to consider the Poisson bracket $\{ H, F_+ \}$, which is a homogeneous polynomial of degree $4$ in $(p_x,p_y)$. Since $F_+$ is a first integral, the bracket vanishes, and hence so do all monomials of that polynomial. In particular, the $p_y^4$-term is
    $$
        e^{\alpha y}\left(
        \frac{x}{(1-x^2)^2}a_2(x) -
        \frac{\alpha}{1-x^2}a_3(x)
        \right)
        p_y^4,
    $$
    which leads to the last equation in~\eqref{eq15}.
    This proves the proposition.
\end{proof}

Note that~\eqref{eq15} is a system of five algebraic equations for five unknown functions $a_0(x), \dots, a_3(x)$ and $h(x)$. This allows us to find all the unknown coefficients of $F_+$ and $F_-$, as well as the explicit form of the metric that depends on $h(x)$. Thus, the existence of the last equation in~\eqref{eq15} is essential since it gives the right number of equations.
The following proposition helps eliminate some trivial solutions of the algebraic system derived in Proposition~\ref{prop2}.

\begin{prop}\label{prop3}
For non-trivial cubic first integrals $F_+$ and $F_-$, the parameter $c_0$ is non-zero. 
\end{prop}
\begin{proof}
    We argue by contradiction and assume that $c_0=0$. Then, by the first equation in~\eqref{eq15}, $a_0(x) \equiv 0$ and
    $$
        F_+ = p_y \, \hat{F}_+
        = L \, \hat{F}_+,
        \qquad
        F_- = -p_y \, \hat{F}_-
        = -L \, \hat{F}_-,
    $$
    where
    \begin{align*}
        \hat{F}_+ &:= e^{\alpha y}\left[
            a_1(x)p_x^2 + a_2(x) p_x p_y + a_3(x) p_y^2 
        \right],
        \\
        \hat{F}_- &:= e^{-\alpha y}\left[
            a_1(x)p_x^2 - a_2(x) p_x p_y + a_3(x) p_y^2 
        \right]
    \end{align*}
    are homogeneous polynomials in momenta of degree two. Since $F_+$, $F_-$, and $L$ are first integrals, so are $\hat{F}_+$ and $\hat{F}_-$, because
    $$
        0 = \{H, F_\pm \} 
        = \{H, \pm L \hat{F}_\pm \}
        = \pm L \{H, \hat{F}_\pm \}.
    $$
    Hence $\{H, \hat{F}_\pm\}=0$ on the open set $L\neq 0$, and therefore everywhere by polynomiality.
    Thus, the cubic integrals are trivial in the sense that they are products of lower-degree integrals. This proves the proposition.
\end{proof}

Let us find the functions $a_0(x), \dots, a_3(x)$ from the system derived in Proposition~\ref{prop2}. First, we solve the last, third and second equations for $a_2$, $a_1$, and $a_0$, respectively. Then, from the fourth equation we find $a_3(x)$. Thus, we obtain
\begin{align}\label{eq16}
    a_0(x) &= \frac{\delta_1 \sqrt{(1-x^2)^3}}{16 \sqrt{2}}
    \frac{
        \left(
            \alpha^2 [h+1]^2 R_2
            - x^2 R_1
        \right)^2
        + 4 x^4 R_2 \left(
            2 c_1 \left[1-x^2\right]
            + 3 c_0
        \right)
    }{
        \alpha\,x^3 (h+1)^4 R_2 \sqrt{R_2}
    },
    \nonumber
    \\
    a_1(x) &= \frac{\delta_1 \sqrt{1-x^2}}{4 \sqrt{2}}
    \frac{
        \alpha^2 (h+1)^2 R_2
        - x^2 R_1
    }{
    x^2 (h+1)^2 \sqrt{R_2}
    },
    \nonumber
    \\
    a_2(x) &= \frac{\delta_1}{2\sqrt{2}}
    \frac{\alpha\sqrt{R_2}}{x\sqrt{1-x^2}},
    \\
    a_3(x) &= \frac{\delta_1}{2\sqrt{2}} \frac{\sqrt{R_2}}{\sqrt{(1-x^2)^3}},
    \nonumber
\end{align}
where
\begin{align}\label{eq17}
    R_1 &:= 
    4 c_2 x^4
    - 4 (c_1 + 2c_2)x^2
    + 3c_0 + 4c_1 + 4c_2,
    \\
    R_2 &:= 
    8 c_3 x^6 
    - 4(c_2+6c_3)x^4 
    + 2(c_1+4c_2+12c_3)x^2
    - (c_0+2c_1+4c_2+8c_3),
    \nonumber
\end{align}
and constant $\delta_1=\pm 1$ simply flips the sign of $F_+$ and $F_-$.

Thus, the coefficients of the cubic integrals are known up to the function $h(x)$ determining the metric. To find all possible Zoll metrics that admit a non-trivial cubic integral, we substitute the recently found $a_0(x)$ into the first equation in~\eqref{eq15}, which leads to
\begin{equation}\label{eq18}
    (1-x^2)^3\mathcal{P}_1(h,x) = 0,
\end{equation}
where $\mathcal{P}_1$ is a polynomial with degrees $(8,24)$ in $(h,x)$. Moreover, it depends only on even powers of $x$. Equation~\eqref{eq18} defines an algebraic curve, and any function $h(x)$ that, when substituted into~\eqref{eq2}, yields a superintegrable Hamiltonian system with a non-trivial cubic integral must satisfy this equation.
Thus, this is a necessary condition. To make it sufficient, one has to ensure that $h(x)$ has all the properties described in Section~\ref{sec11}.

First, we note that
$$
    \mathcal{P}_1(h,0) = \alpha^8 (1+h)^8 \left( c_0 + 2 c_1 + 4 c_2 + 8 c_3 \right)^4 = 0.
$$
The curve must have points on the line $x=0$, which implies
\begin{equation}\label{eq19}
    c_3=-\frac{c_0+2c_1+4c_2}{8}.
\end{equation}
Taking into account~\eqref{eq19},
$$
    \mathcal{P}_1(h,x) = x^8 \mathcal{P}_2(h,x)
$$
for some new polynomial $\mathcal{P}_2$ with degrees $(8, 16)$ in $(h,x)$.

Second, by conditions~\ref{condition:two} and~\ref{condition:three}, we know that $h(0)=0$. Hence, the point $(0,0)$ must belong to the algebraic curve. By condition~\ref{condition:four}, the same applies to the points $(0,1)$ and $(0,-1)$, which leads to
\begin{align*}
    \mathcal{P}_2(0,0) &= (\alpha^2+1)^4 (3 c_0 + 4 c_1 + 4 c_2)^4 = 0,
    \\
    \mathcal{P}_2(0,1) &= \mathcal{P}_2(0,-1) = c_0^4(\alpha^2+1)^3(\alpha^2+9) = 0.
\end{align*}
By Proposition~\ref{prop3}, $c_0 \neq 0$. Since both equations must be satisfied, we find that there are only two possible cases: (i) $\alpha = \pm 3i$ and $3 c_0 + 4 c_1 + 4 c_2 = 0$, or (ii) $\alpha = \pm i$.

In case~(i), the equation takes the form
$$
    x^8 \mathcal{P}_3(h,x) = 0,
$$
where $\mathcal{P}_3$ is a polynomial in $(h,x)$ and
$$
    \mathcal{P}_3(0,0) = 2025 (3c_0+2c_1)^4.
$$
Finally,
$$
    \mathcal{P}_3(h,x) \Big|_{c_1 = -3c_0/2} = 
    9 c_0^4 h (h+2) (3h+2)^3(3h+4)^3 x^8 = 0,
$$
so the trivial solution $h \equiv 0$ is the only admissible one. Therefore, $\alpha \neq \pm 3i$.

In case~(ii), when $\alpha=\pm i$,
$$
    \mathcal{P}_2(h,x) \Big|_{\alpha = \pm i} = \mathcal{P}_4(h,x) \mathcal{P}_5(h,x),
$$
where
\begin{align*}
    \mathcal{P}_4(h,x) &:= 
    \mathcal{Q}_2^2(x) h^4 - 2 x^2 \mathcal{Q}_1(x) \mathcal{Q}_2(x) h^2
    \\
    &- x^4 (1-x^2)^2 \left(
        3 c_0^2
        + 4 c_0 c_1
        + 16 c_0 c_2
        - 4 c_1^2
    \right),
    \\
    \mathcal{P}_5(h,x) &:= 
        h (h+4) \left[
            \left(
                h^2 + 4 h + 8
            \right) \mathcal{Q}_2(x)
            - 2 x^2 \mathcal{Q}_1(x)
        \right] \mathcal{Q}_2(x)
    \\
        & - \left(1-x^2\right) \mathcal{Q}_3(x),
\end{align*}
and
\begin{align}\label{eq20}
    \mathcal{Q}_1(x) &:= -(c_0 + 2 c_1) x^2 + 3 c_0 + 2 c_1,
    \nonumber
    \\
    \mathcal{Q}_2(x) &:= (c_0 + 2 c_1 + 4 c_2) x^4
    - (3 c_0 + 6 c_1 + 8 c_2) x^2
    + 3 c_0 + 4 c_1 + 4 c_2,
    \nonumber
    \\
    \mathcal{Q}_3(x) &:= 
    \left(
        21 c_0^2
        + 92 c_0 c_1
        + 144 c_0 c_2
        + 320 c_1 c_2
        + 100 c_1^2
        + 256 c_2^2
    \right) x^6
    \\
    &- 3 \left(
        39 c_0^2
        + 148 c_0 c_1
        + 208 c_0 c_2
        + 384 c_1 c_2
        + 140 c_1^2
        + 256 c_2^2
    \right) x^4
    \nonumber
    \\
    &+ 24 \left(
        3 c_0
        + 4 c_1
        + 4 c_2
    \right) \left(
        3 c_0
        + 6 c_1
        + 8 c_2
    \right) x^2
    - 16 \left(
        3 c_0
        + 4 c_1
        + 4 c_2
    \right)^2.
    \nonumber
\end{align}
The factor $\mathcal{P}_4(h,x)$ may contain a solution, because the points $(0,0)$, $(0,1)$ and $(0,-1)$ belong to the curve that it defines. As for the second factor,
$$
    \mathcal{P}_5(0,0) = 16 (3 c_0 + 4 c_1 + 4 c_2)^2 = 0.
$$
Proceeding as in case~(i), we obtain
$$
    \mathcal{P}_5(h,x) \Big|_{c_1 = -3c_0/2,\,c_2=3c_0/4} = c_0^2 h (h+2)^2 (h+4) x^8 = 0,
$$
which leads to a trivial solution. Thus, to satisfy conditions
$$
    h(-1)=h(0)=h(1)=0
$$
we have to consider the algebraic curve given by
\begin{equation}\label{eq21}
    \mathcal{P}_4(h,x) = 0
\end{equation}
with parameters $\alpha = \pm i$ and~\eqref{eq19}. Let us note that by Proposition~\ref{prop3} and the homogeneity of $\mathcal{P}_4$ in $c_i$'s, we could rescale the parameters to set \mbox{$c_0 = 1$}. However, we will not do this now and keep the parameter $c_0$ for more convenient comparisons later.

\begin{prop}\label{prop4}
    If $h(x)$ is a non-trivial real branch through $h(1)=0$, then
    $$
        \frac{c_1}{c_0},\frac{c_2}{c_0} \in \mathbb{R}.
    $$
\end{prop}
\begin{proof}
    By Proposition~\ref{prop3}, we can set
    $$
        \sigma_1 := \frac{c_1}{c_0},
        \qquad
        \sigma_2 := \frac{c_2}{c_0},
    $$
    and divide $\mathcal{P}_4(h,x)$ by $c_0^2$. Let $t = 1-x^2$ and $\hat{h}(t)=h^2(x)$. The equation~\eqref{eq21} becomes
    $$ 
        q_2^2(t) \hat{h}^2 - 2 (1-t) q_1(t) q_2(t) \hat{h}
        - (1-t)^2 t^2 \left(
            3
            + 4 \sigma_1
            + 16 \sigma_2
            - 4 \sigma_1^2
        \right) = 0,
    $$
    where
    $$
        q_1(t) :=  
        (1 + 2 \sigma_1) t + 2
        \quad
        \text{and}
        \quad
        q_2(t) := 
        (1 + 2 \sigma_1 + 4 \sigma_2) t^2
        + (1 + 2 \sigma_1) t
        + 1.
    $$
    At $(\hat{h},t)=(0,0)$, the derivative with respect to $\hat{h}$ is
    $$
        -2q_1(0)q_2(0) = -4 \neq 0.
    $$
    Therefore, there is a unique analytic branch through $\hat{h}(0)=0$ by the implicit function theorem. Since $h$ is real-valued, $\hat{h}$ is real-valued, so its Taylor coefficients are real. Write
    $$
        \hat{h}(t) = \tau_1 t^2 + \tau_2 t^3 + O(t^4),
        \qquad
        \tau_1,\tau_2 \in \mathbb{R}.
    $$
    Substituting into the equation gives the first two coefficients:
    $$
        3
        + 4 \sigma_1
        + 16 \sigma_2
        - 4 \sigma_1^2 
        = -4 \tau_1,
        \qquad
        2\tau_2 
        + (6 \sigma_1 + 5) \tau_1
        = 0.
    $$
    Hence, if the branch is non-trivial, so $\tau_1 \neq 0$,
    $$
        \sigma_1 = -\frac{5\tau_1+2\tau_2}{6\tau_1},
        \qquad
        \sigma_2 = -\frac{9\tau_1^3-7\tau_1^2-8\tau_1\tau_2-\tau_2^2}{\tau_1^2}.
    $$
    Thus, $\sigma_1, \sigma_2 \in \mathbb{R}$. This proves the proposition.
\end{proof}

The algebraic curve~\eqref{eq21} is symmetric with respect to both axes, because it depends only on even degrees of $h$ and $x$. We are interested in the odd functions lying on this curve, so, first, we solve~\eqref{eq21} as a quadratic equation in $h^2$, i.\,e.
\begin{equation*}
    h^2_\pm(x) = x^2\frac{\mathcal{Q}_1(x)\mathcal{Q}_2(x) \pm 2\sqrt{c_0\,\mathcal{Q}_2^3(x)}}{\mathcal{Q}_2^2(x)}.
\end{equation*}
Assuming $c_0 \in \mathbb{R}^\times$ by homogeneity of the parameters, the condition~\ref{condition:four} is satisfied if and only if
$$
    h(x) = \delta(x) |x| \sqrt{
        \frac{\mathcal{Q}_1(x) - 2 \sgn c_0 \sqrt{c_0\,\mathcal{Q}_2(x)}}{\mathcal{Q}_2(x)}
    }
$$
for some real function $\delta(x)$ such that $\delta^2(x)=1$. This satisfies condition~\ref{condition:three} if and only if $\delta(-x)=-\delta(x)$, that is
$$
    \delta(x) = \delta_0 \, \mathrm{sgn}\,x
$$
for a constant discrete parameter $\delta_0 = \pm 1$, and therefore
\begin{equation}\label{eq22}
    h(x) = \delta_0 \, x \sqrt{
        \frac{\mathcal{Q}_1(x) - 2 \sgn c_0\sqrt{c_0\,\mathcal{Q}_2(x)}}{\mathcal{Q}_2(x)}
    }.
\end{equation}
Let us now determine when this function satisfies condition~\ref{condition:one}.

\begin{prop}\label{prop5}
    The function $h(x)$ defined in~\eqref{eq22} is real-valued for $x \in [-1,1]$ if and only if
    \begin{align}\label{eq23}
        c_0 (3c_0 + 2c_1) &> 0,
        \nonumber \\
        c_0 (3c_0 + 4c_1 + 4c_2) &> 0,
        \\
        16 c_0 c_2 + (3c_0-2c_1)(c_0+2c_1) &\leq 0.
        \nonumber
    \end{align}
    Furthermore, if
    $$
        16 c_0 c_2 + (3c_0-2c_1)(c_0+2c_1) = 0
    $$
    then $h(x) \equiv 0$.
\end{prop}
\begin{proof}
    By homogeneity of the parameters, we can assume $c_0 \in \mathbb{R}^\times$. Hence, by Proposition~\ref{prop4}, the parameters $c_0$, $c_1$, and $c_2$ are real.
    Let 
    $$
        s_0 := c_0 \sgn c_0 = |c_0|,
        \qquad
        s_1 := c_1 \sgn c_0,
        \qquad
        s_2 := c_2 \sgn c_0,
    $$
    and introduce
    $$
        \mathcal{S}_1(X) := \sgn c_0 \, \mathcal{Q}_1(x),
        \qquad
        \mathcal{S}_2(X) := \sgn c_0 \, \mathcal{Q}_2(x),
        \qquad
        X := x^2 \in [0,1].
    $$
    Then, $c_0\,\mathcal{Q}_2(x) = s_0 \, \mathcal{S}_2(X)$ and
    $$
        \frac{\mathcal{Q}_1(x) - 2 \sgn c_0\sqrt{c_0\,\mathcal{Q}_2(x)}}{\mathcal{Q}_2(x)}
        =
        \frac{\mathcal{S}_1(X) - 2 \sqrt{s_0 \, \mathcal{S}_2(X)}}{\mathcal{S}_2(X)}.
    $$
    Therefore, $h(x)$ is real-valued if and only if
    \begin{align}\label{eq24}
        \mathcal{S}_2(X) > 0
        \qquad
        \text{and}
        \qquad
        \mathcal{S}_1(X) \geq 2 \sqrt{s_0 \, \mathcal{S}_2(X)},
    \end{align}
    for all $X \in [0,1]$, i.\,e.~all the radicands are non-negative (and the denominator is non-zero for $h$ to be continuous). In particular, this implies
    $$
        \mathcal{S}_1(X) = -(s_0+2s_1)X + 3s_0+2s_1 > 0
        \qquad
        \forall X \in [0,1],
    $$
    which is equivalent to
    \begin{equation}\label{eq25}
        3 s_0 + 2 s_1 > 0,
    \end{equation}
    because $\mathcal{S}_1(X)$ is linear and $\mathcal{S}_1(0)=3 s_0 + 2 s_1$ and $\mathcal{S}_1(1)=2s_0>0$.
    Next,
    $$
        \mathcal{S}_2(X) =
        (s_0+2s_1+4s_2)X^2 -
        (3s_0+6s_1+8s_2) X +
        (3 s_0 + 4 s_1 + 4 s_2),
    $$
    so $\mathcal{S}_2(0) > 0$ gives
    \begin{equation}\label{eq26}
        3 s_0 + 4 s_1 + 4 s_2 > 0.
    \end{equation}
    Finally, note that
    $$
    \mathcal{S}_1^2(X) - 4s_0\,\mathcal{S}_2(X)
    =
    (1-X)^2 \, \mathcal{D},
    $$
    where
    \begin{equation}\label{eq27}
        \mathcal{D} = -(3s_0^2-4s_1^2+4s_0s_1+16s_0s_2)
    \end{equation}
    is the   of $\mathcal{S}_2(X)$. Since $\mathcal{S}_1(X) \geq 2 \sqrt{s_0 \, \mathcal{S}_2(X)}$, we must have $\mathcal{D} \geq 0$, that is
    \begin{equation}\label{eq28}
        3s_0^2-4s_1^2+4s_0s_1+16s_0s_2 \leq 0.
    \end{equation}
    This proves the necessity of inequalities~\eqref{eq25},~\eqref{eq26}, and~\eqref{eq28}, that is, inequalities~\eqref{eq23} in the original parameters.
    Note that if $\mathcal{D}=0$, then $h(x) \equiv 0$, so for non-trivial $h(x)$, inequality~\eqref{eq28} becomes strict.

    Conversely, assume~\eqref{eq25},~\eqref{eq26}, and~\eqref{eq28}. In Bernstein form (see, for example,~\cite{Lorentz1953}),
    $$
        \mathcal{S}_2(X) = 
            s_0 X^2 + 
            (3s_0 + 2s_1) X(1-X) +
            (3s_0+4s_1+4s_2) (1-X)^2,
    $$
    where all terms are non-negative, the first coefficient $s_0$ is strictly positive by Proposition~\ref{prop3}, and the last term is strictly positive when the first one vanishes. Hence, $\mathcal{S}_2(X) > 0$ for all $X \in [0,1]$.
    Finally,
    $$
        \mathcal{S}_1^2(X) - 4s_0\,\mathcal{S}_2(X)
        =
        (1-X)^2 \, \mathcal{D} \geq 0.
    $$
    Since $\mathcal{S}_1(X) > 0$, this implies
    $$
        \mathcal{S}_1(X)
        \geq
        2\sqrt{s_0 \, \mathcal{S}_2(X)}.
    $$
    Thus both inequalities in~\eqref{eq24} are satisfied. This proves the proposition.
\end{proof}

Note that for non-trivial solutions, the inequalities given in~\eqref{eq23} must be strict. Indeed, direct substitution of
$$
    c_2 = -\frac{3c_0 + 4c_1}{4}
    \qquad
    \text{or}
    \qquad
    c_2 = -\frac{(3c_0-2c_1)(c_0+2c_1)}{16c_0}
$$
into~\eqref{eq21} together with~\ref{condition:three} and~\ref{condition:four} leads only to the trivial solution $h \equiv 0$. 

The next statement concerns condition~\ref{condition:five} for function~\eqref{eq22}.

\begin{prop}\label{prop6}
    The function $h(x)$ defined in~\eqref{eq22} satisfies
    $$
        h(x) + 1 > 0
    $$
    for all $x \in [-1,1]$ provided the conditions~\eqref{eq23} on the parameters are met.
\end{prop}
\begin{proof}
    Since $h(x)$ is an odd function, we need to prove that $|h(x)|<1$ for $x \in [-1,1]$. Equivalently,
    $$
        |h(x)|^2 = x^2 \left(
            \frac{\mathcal{Q}_1(x) - 2 \sgn c_0 \sqrt{c_0\,\mathcal{Q}_2(x)}}{\mathcal{Q}_2(x)}
        \right)
        < 1
    $$
    or, using the notation from the proof of Proposition~\ref{prop5},
    \begin{equation}\label{eq29}
        \mathcal{S}_2(X) 
        - X\mathcal{S}_1(X) 
        + 2 X \sqrt{s_0\,\mathcal{S}_2(X)}
        > 0.
    \end{equation}
    Let
    $$
        C_1 := 3 s_0 + 2 s_1
        \qquad
        \text{and}
        \qquad
        C_2 := 3s_0+4s_1+4s_2.
    $$
    By Proposition~\ref{prop5}, both $C_1$ and $C_2$ are positive.
    In the Bernstein basis,
    \begin{align*}
        \mathcal{S}_1(X) &= 2 s_0 X + C_1 (1-X),
        \\
        \mathcal{S}_2(X) &= 
            s_0 X^2 + 
            C_1 X(1-X) +
            C_2 (1-X)^2.
    \end{align*}
    Clearly, for $C_1 > 0$, $C_2 > 0$, and $X\in[0,1]$, 
    \begin{equation}\label{eq30}
        \mathcal{S}_2(X) \geq s_0 X^2
        \quad
        \Rightarrow
        \quad
        \sqrt{s_0\,\mathcal{S}_2(X)} \geq s_0 X.
    \end{equation}
    In the Bernstein basis,~\eqref{eq29} becomes
    $$ 
        C_2 (1-X)^2
        - s_0 X^2 
        + 2 X \sqrt{s_0\,\mathcal{S}_2(X)}
        > 0.
    $$
    Finally, by~\eqref{eq30},
    $$
        C_2 (1-X)^2
        - s_0 X^2 
        + 2 X \sqrt{s_0\,\mathcal{S}_2(X)}
        \geq
        C_2 (1-X)^2 
        + s_0 X^2 > 0,
    $$
    which implies $|h(x)| < 1$ and proves the proposition.
\end{proof}

To prove that function~\eqref{eq22} is indeed a solution leading to Zoll metrics, we only need to check that it satisfies condition~\ref{condition:two}.

\begin{prop}\label{prop7}
    The function $h(x)$ defined in~\eqref{eq22} is smooth, namely
    $$
        h \in C^\infty[-1,1]
    $$
    for parameter values given in~\eqref{eq23}.
\end{prop}
\begin{proof}
    Multiplying the numerator and denominator of~\eqref{eq22} by 
    $$
        \sqrt{\mathcal{Q}_1 + 2 \sgn c_0\sqrt{c_0\,\mathcal{Q}_2}},
    $$    
    one can write $h(x)$ as
    $$
        h(x) = \delta_0\sqrt{
            \mathcal{D}
            } 
            \frac{x(1-x^2)}{
            \sqrt{
            \mathcal{Q}_2(\mathcal{Q}_1 + 2 \sgn c_0\sqrt{c_0\,\mathcal{Q}_2})
            }},
    $$
    where $\mathcal{D} \geq 0$ is defined by~\eqref{eq27}. Now, the numerator is a polynomial, and, in the notation from the proof of Proposition~\ref{prop5}, the new denominator is
    $$
        \sqrt{
            \mathcal{S}_2(X)
            \left(
            \mathcal{S}_1(X) 
            + 2 \sqrt{s_0\,\mathcal{S}_2(X)}
            \right)
        },
    $$
    which is a strictly positive smooth function on $X \in [0,1]$. Hence, $h$ is infinitely differentiable as a composition of $C^\infty$ functions. 
    
    This proves the proposition.
\end{proof}

Hence, we have shown that $h(x)$ defined in~\eqref{eq22} satisfies necessary and sufficient conditions~\ref{condition:one}--\ref{condition:five} for the parameter values $\delta_0 = \pm 1$ and~\eqref{eq23}. Consequently, according to~\cite[Corollary~4.16]{Besse1978}, the metric~\eqref{eq1} extends smoothly to $\mathbb{S}^2$ and all its geodesics are closed. Finally, 
by changing the sign of $x$, we can set $\delta_0=+1$, and by rescaling the parameters $c_0, \dots, c_3$, we can set $c_0=1$, which leads to the final solution~\eqref{eq6} with two essential parameters $\sigma_1 := c_1/c_0$ and $\sigma_2 := c_2/c_0$.

Let us note that $\mathcal{Q}_i$'s defined in~\eqref{eq20} and $Q_i$'s used in Theorem~\ref{t1} are related as
$$
    c_0\,Q_1(x) = \mathcal{Q}_1(x)
    \qquad
    \text{and}
    \qquad
    c_0\,Q_2(x) = \mathcal{Q}_2(x)
$$
after the change of parameters $(c_1, c_2) \to (c_0\,\sigma_1, c_0\,\sigma_2)$ in $\mathcal{Q}_i$'s.
Moreover, substituting~\eqref{eq19} into the polynomials $R_1$ and $R_2$ defined in~\eqref{eq17} leads to the following relations:
$$
    R_1(x) = \mathcal{Q}_2(x) + x^2 \mathcal{Q}_1(x),
    \qquad
    R_2(x) = -x^2 \mathcal{Q}_2(x).
$$
The first integrals $F_+$ and $F_-$ are complex-valued functions. To obtain real cubic integrals, we substitute~\eqref{eq16} with $\delta_1 = \sgn x$ and polynomials $R_1$ and $R_2$ as above into the definitions~\eqref{eq10} and~\eqref{eq11}, and then define
$$
    \mathcal{F}_1 := - \sqrt{2} (F_+ + F_-)
    \qquad
    \text{and}
    \qquad
    \mathcal{F}_2 := i \sqrt{2} (F_+ - F_-).
$$
These functions are real-valued and cubic in $(p_x,p_y)$. More explicitly,
\begin{align*}
    \mathcal{F}_1 &= 
    \sqrt{1-x^2} \left[
    \cos y \left(
        \mathcal{A}_0 p_x^{3}
        + 
        \mathcal{A}_2 p_x p_y^2
    \right)
    + \sin y \left(
        \mathcal{A}_1 p_x^{2} p_y
        + 
        \mathcal{A}_3 p_y^3
    \right)
    \right],
    \\
    \mathcal{F}_2 &= 
    \sqrt{1-x^2} \left[
    \sin y \left(
        \mathcal{A}_0 p_x^{3}
        + 
        \mathcal{A}_2 p_x p_y^2
    \right)
    - \cos y \left(
        \mathcal{A}_1 p_x^{2} p_y
        + 
        \mathcal{A}_3 p_y^3
    \right)
    \right],
\end{align*}
where
\begin{align*}
    \mathcal{A}_0(x) &:= \frac{1-x^2}{(h+1)^3},
    &
    \mathcal{A}_1(x) &:= \frac{
        x^2 Q_1
        - Q_2 h(h+2) 
    }{
        2x(h+1)^2\sqrt{Q_2}
    },
    \\
    \mathcal{A}_2(x) &:= \frac{\sqrt{Q_2}}{1-x^2},
    &
    \mathcal{A}_3(x) &:= \frac{x \sqrt{Q_2}}{(1-x^2)^2}.
 \end{align*}
By direct substitution, one can verify that $\{H,\mathcal{F}_1\}=\{H,\mathcal{F}_2\}=0$. 

\begin{prop}\label{prop8}
    The first integrals $\mathcal{F}_1$ and $\mathcal{F}_2$ are regular on the sphere $\mathbb{S}^2$.
\end{prop}
\begin{proof}
    By Proposition~\ref{prop7}, $h \in C^\infty[-1,1]$.
    By Proposition~\ref{prop6}, $h+1>0$ for all $x \in [-1,1]$, so the functions $\mathcal{A}_0$, $\mathcal{A}_2$, and $\mathcal{A}_3$ are regular for $x \in (-1,1)$. Recall that
    $$
        h(x) = x E(x),
        \qquad
        E(x) := \sqrt{
            \mathcal{D}
            } 
            \frac{1-x^2}{
            \sqrt{
            Q_2(Q_1 + 2 \sqrt{Q_2})
            }},
    $$
    where $E$ is the even part of the function $h$. After this substitution into $\mathcal{A}_1$, the apparent singularity at $x=0$ cancels, so
\begin{align*}
    \mathcal{A}_1 &= \frac{
        x \, Q_1
        - Q_2 E(h+2) 
    }{
        2(h+1)^2\sqrt{Q_2}
    },
\end{align*}
is also regular for $x \in (-1,1)$. Thus, it remains to check the coordinate singularities at $x=\pm 1$. In local coordinates
$$
    X = \sqrt{1-x^2} \cos y,
    \qquad
    Y = \sqrt{1-x^2} \sin y,
$$
the cubic integrals are
\begin{align*}
    \mathcal{F}_1 &= 
    X \left(
        \mathcal{A}_0 p_x^{3}
        + 
        \mathcal{A}_2 p_x p_y^2
    \right)
    + Y \left(
        \mathcal{A}_1 p_x^{2} p_y
        + 
        \mathcal{A}_3 p_y^3
    \right),
    \\
    \mathcal{F}_2 &=
    Y \left(
        \mathcal{A}_0 p_x^{3}
        + 
        \mathcal{A}_2 p_x p_y^2
    \right)
    - X \left(
        \mathcal{A}_1 p_x^{2} p_y
        + 
        \mathcal{A}_3 p_y^3
    \right).
\end{align*}
The new canonical momenta are
$$
    p_x = - \frac{x}{X^2+Y^2}(X p_X + Y p_Y),
    \qquad
    p_y = X p_Y - Y p_X,
$$
and substituting this, one obtains
\begin{align*}
    \mathcal{F}_1 &= 
    - x \sqrt{Q_2} p_X (p_X^2 + p_Y^2) 
    - \frac{x \, \mathcal{B}_1}{R} p_X (X p_X + Y p_Y)^2
    \\
    &+
    \frac{x \, \mathcal{B}_2}{R^2} Y (X p_X + Y p_Y)^2 (X p_Y - Y p_X),
    \\
    \mathcal{F}_2 &= - x \sqrt{Q_2} p_Y (p_X^2 + p_Y^2)
    - \frac{x \, \mathcal{B}_1}{R} p_Y (X p_X + Y p_Y)^2
    \\
    &-
    \frac{x \, \mathcal{B}_2}{R^2} X (X p_X + Y p_Y)^2 (X p_Y - Y p_X),
\end{align*}
where
$$
    \mathcal{B}_1 := 
        \frac{x^2}{(h+1)^3} - \sqrt{Q_2},
    \quad
    \mathcal{B}_2 :=
        \frac{x^2-h\sqrt{Q_2}}{(h+1)^2} - \frac{x^2}{(h+1)^3},
    \quad
    R := X^2+Y^2.
$$
Since at $x = \pm 1$, one has
$$
    \mathcal{B}_1 = O(R),
    \qquad
    \mathcal{B}_2 = O(R^2),
$$
both integrals extend smoothly to the poles. This proves the proposition.
\end{proof}

Let us recall that for a compact Liouville surface with Gaussian curvature not positive constant, any two non-trivial quadratic first integrals are dependent (modulo the energy); see~\cite{Kiyohara1991}. The curvature for the metric $g$ is
$$
    K(x) = \frac{h(x) + 1 - x h'(x)}{(h(x)+1)^3},
$$
which is non-constant for $h \not\equiv 0$.
Thus, the only possible integrals of degrees lower than $3$ are $H$ and $L$. Therefore, the integrals $\mathcal{F}_1$ and $\mathcal{F}_2$ can be trivial only if they are functionally dependent on $H$ and $L$. However, both $H$ and $L$ do not depend on the angle $y$, while both $\mathcal{F}_1$ and $\mathcal{F}_2$ do. Thus, neither $\mathcal{F}_1$ nor $\mathcal{F}_2$ are functionally dependent on lower-degree integrals and therefore are non-trivial.

This is the only family of solutions for $\alpha \neq 0$ that admits a non-trivial integral cubic in the momenta. To complete the proof of Theorem~\ref{t1}, we now need to consider the case where the operator $\mathcal{L}$ does not have non-zero eigenvalues.

\subsection{The Case of Zero Eigenvalues}\label{sec22}

We now show that if the map $\mathcal{L}$ admits only the eigenvalue~$0$, then the metric $g$ cannot be a rotationally symmetric Zoll metric on $\mathbb{S}^2$. 

\begin{prop}\label{prop9}

If the map $\mathcal{L}$ admits only the eigenvalue~$0$, then the function $h(x)$ in \eqref{eq1} cannot be odd.
    
\end{prop}

\begin{proof}
    Consider $F$ of the form~\eqref{eq4}. Since~$0$ is an eigenvalue, there exist constants $b_0$ and $b_1$ such that 
    \begin{equation}\label{eqZE1}
        \{ F, L \}= b_0 \, L^3 + b_1 \, HL.
    \end{equation} 

    Directly substituting ~\eqref{eq2}, \eqref{eq3}, and~\eqref{eq4} into $\{ F, L \} - b_0 \, L^3 - b_1 \, HL\equiv 0$ makes its left-hand side a homogeneous polynomial of degree $3$ in $(p_x, p_y)$. Its monomials have to vanish simultaneously, which gives us a system of four PDEs in the unknowns $A_i(x,y)$ for $i=0,\dots, 3$. It turns out to be a first order linear ODE system in the variable $y$, hence we can solve it and we get the following simplified expressions:
    \begin{align*}
        A_0(x,y) & = a_0(x) \\
        A_1(x,y) & = a_1(x)+\frac{b_1 \left(x^2-1\right)}{2 (h(x)+1)^2}y\\
        A_2(x,y) & = a_2(x) \\
        A_3(x,y) & = a_3(x)+\frac{-2 \,b_0 \,x^2+b_1+2 b_0}{2 \left(x^2-1\right)}y.
    \end{align*}

    We substitute these expressions in the cubic first integral $F$ and consider the condition $\{H,F\}\equiv0$. Its left-hand side is a homogeneous polynomial of degree 4 in momenta and we consider again the system of equations generated by equating its monomials to zero. The variable $y$ cancels and the coefficient of $p_y^4$ gives an algebraic linear expression in the variable $a_2(x)$, from which we get 
    $$ a_2(x)= -\frac{b_1+2 b_0 \left(1-x^2\right)}{2 x}.$$
    We substitute $a_2(x)$ in the remaining equations and get a system of four first order ODEs that are separable in the four unknown functions $a_0(x), a_1(x),$ $a_3(x), h(x)$.
    \begin{align*}
    & 2 x a_1(x) (h(x)+1)^2+\left(x^2-1\right)^3 a_3'(x)=0\\[8pt]
    & 6 x^3 a_0(x) (h(x)+1)^3+2 b_0 \left(x^2-1\right)^3 \left(x \left(x^2-1\right) h'(x)+h(x)+1\right)+\\
    & +b_1 \left(x^2-1\right)^2 \left(\left(3 x^2-1\right) (h(x)+1)-x \left(x^2-1\right) h'(x)\right)=0\\[8pt]
    &2 a_1(x) \left(\left(x^2-1\right) h'(x)-x (h(x)+1)\right)+\left(x^2-1\right) (h(x)+1) a_1'(x)=0\\[8pt]
    &3 a_0(x) \left(\left(x^2-1\right) h'(x)-x (h(x)+1)\right)+\left(x^2-1\right) (h(x)+1) a_0'(x)=0
    \end{align*}

    We can solve it and get: 
    \begin{align*}
        a_0(x) & = \frac{b_2 \left(1-x^2\right)^{3/2}}{(h(x)+1)^3}, \qquad a_1(x)  = \frac{b_3 \left(1-x^2\right)}{(h(x)+1)^2}, \qquad a_3(x) = \frac{b_3}{1-x^2}+b_4,  \\
        \\
        h(x) & = -1+\left(1-x^2\right) \left(\frac{2 \,b_2\, x}{b_1 \left(1-x^2\right)^{3/2}}+\frac{ b_5 \, x}{\left(b_1+2 b_0 \left(1-x^2\right)\right){}^{3/2}}\right)
    \end{align*}
    where $b_i\in\mathbb{R}$.
    
    The function $h(x)$ is not an odd function and we are done. 
\end{proof}
This proposition concludes the proof of Theorem~\ref{t1}.

\section{Comparison with Known Results}\label{sec3}

We compare our main result, Theorem \ref{t1}, with the result found by Valent, Duval, and Shevchishin in \cite{ValentDuvalShevchishin2015} and summarised by Valent in \cite{Valent2014}. They claim to have determined all superintegrable systems with a non-trivial cubic polynomial first integral globally defined on $\mathbb{S}^2$. More specifically, comparing the corresponding Casimir, namely our expression \eqref{eq12} and the formula (27) of \cite{Valent2014}, we find a relation between the free parameters $\sigma_1,\sigma_2$ that appear in Theorem \ref{t1} and their free parameters $l,m$. For simplicity, we rescale the parameters $c_0,\dots,c_2$
with $c_0=8$ instead of $c_0=1$ and from \eqref{eq23} we get the conditions
\begin{equation}\label{eq7_t}
    c_1 > -12
    \quad
    \text{and}
    \quad
    -\left(c_1+6\right) 
    < c_2 <
    \frac{(c_1+4)(c_1-12)}{32}
\end{equation}
which are the analog of \eqref{eq7}, i.e. necessary and sufficient conditions on the parameters for a rotationally symmetric Zoll metric whose geodesic flow admits a non-trivial cubic integral.

The integral $P$ in their paper has the form
\begin{equation*}
    P = \left(2H\right)^3+ \left(-\frac{3l+m}{l+m}\right) \left(2H\right)^2L^2+ \left(\frac{3l^2+2lm-1}{\left(l+m\right)^2}\right) \left(2H\right)L^4+ \left(-\frac{l^2-1}{\left(l+m\right)^2}\right) L^6
\end{equation*}
and their admissible parameters are $l\ge-1$ and $m>1$.

Hence, our parameters can be expressed in terms of $l$ and $m$ as:
\begin{equation*}
    c_1 =- \frac{4(3l+m)}{l+m},
    \quad
    c_2 = \frac{2(3l^2+2lm-1)}{\left(l+m\right)^2},
    \quad
    c_3 = -\frac{l^2-1}{\left(l+m\right)^2}. 
\end{equation*}
In terms of $l$ and $m$ the linear condition \eqref{eq19} with $c_0=8$ is automatically satisfied. Assuming $l+m>0$, we find a diffeomorphism between the parameters $l,m$ and $c_1,c_2$, namely: 
\begin{equation*}
    l = - \frac{4+c_1}{\sqrt{(c_1+4)(c_1-12)-32c_2}},
    \quad
    m = \frac{12+c_1}{\sqrt{(c_1+4)(c_1-12)-32c_2}},
\end{equation*}
which is well defined in the region of admissible parameters \eqref{eq7_t}.

We plot their region of admissible parameters in the coordinates $c_1$ and $c_2$ in Fig.~\ref{fig2} and compare to our result, which gives a larger region. 
\begin{figure}[ht]
    \centering
    \includegraphics[width=0.6\linewidth]{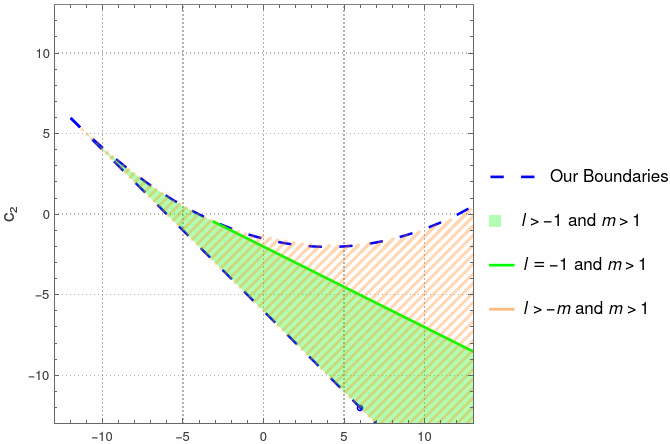}
    \caption{Comparison of regions of admissible parameters}
    \label{fig2}
\end{figure}
Specifically their region is $\{l>-1,m>1\}$ while the whole region of admissible parameters we found is $\{l>-m,m>1\}$. 

\subsection{A Special Case of the N-degree Valent System}

The two-parameter family of rotationally symmetric Zoll metrics with cubic integrals found by Valent in~\cite{Valent2021} covers the whole region of admissible parameter values. Let us compare the parameters. After dividing the Hamiltonian from~~\cite{Valent2021} by $2$ to set a proper energy level, one can check that the linear condition~\eqref{eq19} is satisfied. The integral $P$ then becomes
$$
    P = 
    8H^3 
    - 4(m_1 + m_2 + 1) H^2 L^2 
    + 2(m_1 + m_2 + m_1 m_2) H L^4
    - m_1 m_2 L^6,
$$
where $m_1, m_2 < 1$ are the real parameters of the superintegrable system found in~\cite{Valent2021}. Hence, up to the scaling of $P=F_+ F_-$ by $2^{-3}$ and the permutation of $m_1$ and $m_2$,
$$
    m_1 = -\frac{2\sigma_1 + 1 + \sqrt{\mathcal{D}}}{2}
    \qquad
    \text{and}
    \qquad
    m_2 = -\frac{2\sigma_1 + 1 - \sqrt{\mathcal{D}}}{2},
$$
where
$$
    \mathcal{D} = 4\sigma_1^2 - 4\sigma_1 - 16\sigma_2 - 3
$$
is the discriminant of $Q_2(x^2)$. The conditions on the parameters $m_1$ and $m_2$ are precisely the conditions~\eqref{eq7} on the parameters $\sigma_1$ and $\sigma_2$, that is,
$$
    m_1 < 1
    \qquad
    \text{and}
    \qquad
    m_2 < 1
$$
if and only if
$$
    2\sigma_1 + \sqrt{\mathcal{D}} > -3,
    \qquad
    2\sigma_1 - \sqrt{\mathcal{D}}> -3,
    \qquad
    \text{and}
    \qquad
    \mathcal{D} \geq 0.
$$
Adding and subtracting the first two inequalities, we obtain
$$
    2 \sigma_1 > -3
    \qquad
    \text{and}
    \qquad
    \mathcal{D} > 0.
$$
The final inequality
$$
    4 \sigma_1 + 4 \sigma_2 + 3 > 0
$$
is equivalent to
$$
    (m_1-1)(m_2-1)>0,
$$
which holds for $m_1<1$ and $m_2<1$.
Thus, in \cite{Valent2021} does give the whole family, but there he does not claim made that this is the whole family. Our proof shows that these indeed are all Funk metrics with linear and cubic integral.

\bigskip
\noindent
\textbf{\large Acknowledgements}
\medskip

\noindent
We thank the mathematical research institute MATRIX in Australia and La~Trobe University, where parts of this research were carried out. G.\,P.\,P.~and H.\,R.\,D. also acknowledge support from the Australian Research Council through DP210100951.

\bibliography{zoll_funk_superintegrability.bib}
\bibliographystyle{plain}

\end{document}